\pdfoutput=1

\documentclass[prl,aps,showpacs,twocolumn,twoside,superscriptaddress,floatfix]{revtex4-1}
\usepackage{amsmath,amssymb,latexsym,revsymb,theorem,verbatim,color}

\newtheorem{definition}{Definition}

\newtheorem{lemma}[definition]{Lemma}

\newtheorem{theorem}[definition]{Theorem}

\def\squareforqed{\hbox{\rlap{$\sqcap$}$\sqcup$}}
\def\qed{\ifmmode\squareforqed\else{\unskip\nobreak\hfil
\penalty50\hskip1em\null\nobreak\hfil\squareforqed
\parfillskip=0pt\finalhyphendemerits=0\endgraf}\fi}
\def\endenv{\ifmmode\;\else{\unskip\nobreak\hfil
\penalty50\hskip1em\null\nobreak\hfil\;
\parfillskip=0pt\finalhyphendemerits=0\endgraf}\fi}
\newenvironment{proof}{\noindent \textbf{{Proof.~} }}{\qed}

\def\r{\rho}

\newcommand{\nc}{\newcommand}
\nc{\bra}[1]{\langle#1|} \nc{\ket}[1]{|#1\rangle} \nc{\proj}[1]{|
#1\rangle\!\langle #1 |} \nc{\ketbra}[2]{|#1\rangle\!\langle#2|}
\nc{\braket}[2]{\langle#1|#2\rangle} \nc{\dg}{\dagger}
\nc{\norm}[1]{\lVert#1\rVert} \nc{\abs}[1]{|#1|}
\nc{\lar}{\leftarrow} \nc{\rar}{\rightarrow} \nc{\ox}{\otimes}
\nc{\op}[2]{|#1\rangle\langle#2|} \nc{\ip}[2]{\langle#1|#2\rangle}

\nc{\cA}{{\cal A}} \nc{\cB}{{\cal B}} \nc{\cC}{{\cal C}}
\nc{\cD}{{\cal D}} \nc{\cE}{{\cal E}} \nc{\cF}{{\cal F}}
\nc{\cG}{{\cal G}} \nc{\cH}{{\cal H}} \nc{\cI}{{\cal I}}
\nc{\cJ}{{\cal J}} \nc{\cK}{{\cal K}} \nc{\cL}{{\cal L}}
\nc{\cM}{{\cal M}} \nc{\cN}{{\cal N}} \nc{\cO}{{\cal O}}
\nc{\cP}{{\cal P}} \nc{\cR}{{\cal R}} \nc{\cS}{{\cal S}}
\nc{\cT}{{\cal T}} \nc{\cX}{{\cal X}} \nc{\cZ}{{\cal Z}}

\begin{document}

\title{Detecting Multipartite Classical States and their Resemblances}

\author{Lin Chen}
\address{Centre for Quantum Technologies, National University of
Singapore, 3 Science Drive 2, 117542, Singapore}

\author{Eric Chitambar}
\affiliation{Department of Physics and Department of Electrical \& Computer Engineering,
University of Toronto, Toronto, Ontario, M5S 3G4, Canada}

\author{Kavan Modi}
\address{Centre for Quantum Technologies, National University of
Singapore, 3 Science Drive 2, 117542, Singapore}

\author{Giovanni Vacanti}
\address{Centre for Quantum Technologies, National University of
Singapore, 3 Science Drive 2, 117542, Singapore}

\begin{abstract}

We study various types of multipartite states lying near the quantum-classical boundary.  The class of so-called classical states are precisely those in which each party can perform a projective measurement to identify a locally held state without disturbing the global state, a task known as non-disruptive local state identification (NDLID).  We introduce a new class of states called generalized-classical states which allow for NDLID when the most general quantum measurements are permitted.  A simple analytic method as well as a physical criterion are presented for detecting whether a multipartite state is classical.  To decide whether a state is generalized-classical, we provide a semi-definite programming algorithm which can be adapted for use in other unrelated contexts such as signal processing.

\end{abstract}

\date{\today}

\pacs{03.67.-a, 03.65.Ud, 03.67.Mn}

\maketitle

\emph{Introduction.}---There are many ways in which composite
quantum systems can exhibit non-classical properties.  The
correlations between entangled states have generated some of the
most puzzling paradoxes in quantum theory;
however even unentangled, or \textit{separable} states, possess
correlations that \textit{cannot} be simulated by classical systems
and thus defy our intuition.  Recently, much interest has been raised concerning the properties of these non-classical correlations with applications 
to a variety of fields \cite{Oppenheim-2002a,Groisman-2005a,Datta-2008a,Lanyon-2008a,Piani-2008a,Datta-2009a,Brodutch-2010a,Soares-Pinto-2010a}.  Of particular note is the DQC1 quantum computation model which runs exponentially faster than its best-known classical counterpart by using a highly mixed state possessing quantum correlations but no entanglement \cite{Datta-2008a,Lanyon-2008a}.  This supports a hypothesis that non-classical correlations are a more fundamental resource than entanglement in quantum computing.

In light of this, several measures have been designed to isolate and quantify 
precisely the non-classical nature of a quantum state such as quantum 
discord \cite{Ollivier-2001a}, quantum deficit \cite{Oppenheim-2002a}, measurement induced 
disturbance \cite{Luo-2008a}, and similar quantities \cite{Groisman-2005a,Devi-2008a,Modi-2010a}.  One common feature of all these measures is 
that they vanish for \textit{fully classical} states, i.e. those in which the shared correlations among all the parties can be simulated on a classical system.  Thus, any such measure can be interpreted as quantifying how far away a given state is from the classical-quantum border, even within the class of separable states.

In this Letter, we take an alternative approach to the sharpening of the quantum-classical boundary region; instead of grouping states in this region according to some numerical distance away from the set of classical states, we identify a state as ``nearly'' classical if it possesses a well-defined trace of some purely classical property.  Specifically, we address the following two questions: (i) in what physical ways can general quantum states resemble classical states, and (ii) how can one detect whether a given state is classical or at least resemblant to one in the sense of question (i)?  One answer to the first question, which we investigate below, involves a state's ability to undergo \textit{non-disruptive local state identification} (NDLID).  In the remainder of this letter, we will first give a precise description of NDLID and characterize the states which exhibit this property.  NDLID capable states are found to occupy a measure zero volume of state space and belong to the class of so-called minimal length separable states.  After that, we will proceed to answer question (ii) by providing computational and experimental methods for deciding whether or not a given multipartite state is classical or even just similar to one in its ability for NDLID.  Our detection algorithm can be efficiently implemented which differs drastically from the best known methods of detecting separability.

\emph{NDLID and a Hierarchy of Separable States.}---As a motivating example, consider the fully classical state $\rho=\frac{1}{2}(\op{00}{00}+\op{11}{11})$.  Each party can perform a projective measurement in the computational basis and learn his/her local state to be either $\ket{0}$ or $\ket{1}$.  When these results are not recorded or kept secret, the post-measurement state is still $\rho$, and the parties have thus identified their state without perturbing the overall state.  The ability for each party to perform such an information-gathering process without failure is not particular to this example but, in fact, completely characterizes the set of fully classical states \cite{Ollivier-2001a}.  As a result, the possibility for a given state to undergo some sort of NDLID can be regarded as a signature of ``classicalness.''

In general, we will say a state $\rho$ allows for NDLID by party $k$ if there exists a decomposition $\rho=\sum_ip_i\rho^{(\overline{k})}_i\otimes\op{\phi_i^{(k)}}{\phi_i^{(k)}}$ and local measurement $\{M_i^{(k)}\}_{i=1...n}$ with $\sum^n_{i=1}M^{(k)\dagger}_iM^{(k)}_i\leq I^{(k)}$ such that 
\begin{equation}
\label{Eq:localdist} M^{(k)}_i\op{\phi_j^{(k)}}{\phi_j^{(k)}}M^{(k)\dagger}_i=\lambda\delta_{ij}\op{\phi^{(k)}_j}{\phi^{(k)}_j}
\end{equation} for some $0<\lambda\leq 1$.  Upon outcome $i$, party $k$ can then conclude that his/her system is in state $\ket{\phi_i^{(k)}}$ among the ensemble $\{\ket{\phi_j^{(k)}}\}$, while the rest of the system is in state $\rho^{(\overline{k})}_i$.  Furthermore, it can easily be seen that under the action of this measurement, the global state remains invariant: $\sum^n_{i=1}(I^{(\overline{k})}\otimes M^{(k)}_i)\rho(I^{(\overline{k})}\otimes M^{(k)\dagger}_i)=\lambda \rho$.  

From Eq. \eqref{Eq:localdist}, it immediately follows that the task of NDLID is equivalent to unambiguous state discrimination among the states $\ket{\phi_j^{(k)}}$ with a post-selection rate of $\lambda$.  A well-known necessary and sufficient condition for accomplishing this feat is that the $\ket{\phi_j^{(k)}}$ are linearly independent \cite{Chefles-1998b}.  In this case, the measurement operators take the form $M^{(k)}_i=\op{\phi^{(k)}_{i_k}}{\phi^{(k)\perp}_{i_k}}$
where $\ip{\phi^{(k)}_{j_k}} {\phi^{(k)\perp}_{i_k}}=\delta_{ij}\lambda$ for some
$0<\lambda\leq 1$.  Furthermore, we have $\lambda=1$ if and only if the $\ket{\phi_j^{(k)}}$ are orthogonal and the NDLID can be performed by a complete projective measurement.  These facts motivate the following classifications of multipartite separable states.  
\begin{definition}
\label{defn:main} Let
$\{\ket{\phi(\vec{i})}\}=\{\ket{\phi^{(1)}_{i_1}\phi^{(2)}_{i_2}\dots\phi^{(N)}_{i_N}}\}$
denote a product state basis. {\begin{itemize}
\item[a.] A multipartite state $\rho$ is called \sl{separable} if it is diagonal in some product state basis; i.e.
\[\rho = \sum_{\vec{i}} p_{\vec{i}} \proj{\phi(\vec{i})},\]
\item[b.] The state $\rho$ is called \sl{generalized-classical} for the $k^{th}$ party if it is diagonal in some product state basis in which the states $\{\ket{\phi^{(k)}_{i_k}}\}$ are linearly independent.
\item[c.] The state $\rho$ is called \sl{classical} for the $k^{th}$ party if it is diagonal in some product state basis in which the states $\{\ket{\phi^{(k)}_{i_k}}\}$ are orthogonal.
\item[d.] The state $\rho$ is called \sl{fully generalized-classical} or \sl{fully classical} if it is diagonal in some product state basis in which statements b or c are true respectively for all parties.
\end{itemize}}
\end{definition} 
\noindent From the discussion preceding Definition \ref{defn:main}, generalized-classical states are nearly classical in the following sense:
\begin{quotation}
\noindent \textit{A state is classical (resp. generalized-classical) with respect to party $k$ iff party $k$ can perform NDLID by a projective (resp. generalized) measurement.}
\end{quotation}

There exists an even broader class of separable states still
hovering close to the quantum-classical border.  An $N$-partite state
$\rho$ of rank $r$ will be called a \textit{minimal length separable
state} if it has a decomposition $\rho=\sum_{i=1}^r\lambda_i
\op{\phi^{(1)}_1\cdots \phi^{(N)}_i}{\phi^{(1)}_1\cdots \phi^{(N)}_i}$ \cite{DiVincenzo-2000a}. It is
quite easy to see from the following lemma that any fully
generalized-classical state is also a minimal length separable state.

\begin{lemma}
For some multi-index $(i_1,...,i_N)$, if up to repetition of states the
$\ket{\phi^{(j)}_{i_j}}$ are linearly independent for all parties $j$, then
the product states $\ket{\phi^{(1)}_{i_1}\cdots\phi^{(N)}_{i_N}}$ are also
linearly independent.
\end{lemma}

By this lemma and Definition \ref{defn:main}, if $\rho$ is
fully generalized-classical, it has a decomposition
$\rho=\sum_{i=1}^\delta\op{\phi^{(1)}_{i_1}\cdots\phi^{(N)}_{i_N}}{\phi^{(1)}_{i_1}\cdots\phi^{(N)}_{i_N}}$
with $\delta\geq r$ and each $\ket{\phi^{(1)}_{i_1}\cdots\phi^{(N)}_{i_N}}$
linearly independent.  This last property implies that $r=\delta$ and
so we see that each fully generalized-classical state is a minimal length
state.  Furthermore, in the bipartite case, if a state is generalized-classical with respect to just one of the parties, it will be of minimal length.  The following chain of inclusions summarizes the main
parsings described in this letter:
\begin{center}
separable $\supset$ minimal length $\supset$ fully generalized-classical
$\supset$ fully classical $\supset$ product.
\end{center}
Here, product states refer to states of the form
$\rho=\rho_1\otimes\cdots\otimes \rho_N$. 

There are two reasons to consider minimal length states as also
lying near the quantum-classical border.  First, it is known that
only \textit{non}-minimal length states constitute the opposite end
of the spectrum at the separable/non-separable boundary
\cite{DiVincenzo-2000a}.  While this alone does not imply a
closeness between minimal length and classical states, such an
interpretation becomes further justified when considering the volumes of
each set in state space.  Separable states possess a nonzero volume
\cite{Zyczkowski-1998a} while minimal length states are of measure
zero \cite{Lockhart-2000a}.  This final point has an even greater
relevance to our discussion since it implies that fully
generalized-classical states are also of measure zero.  In other words,
nearly all multipartite quantum states lack the property of
non-disruptive local state identification.  Also note that this
provides an alternative proof for the result in Ref. \cite{Ferraro-2010a}
which shows a generic state to have a nonzero discord (i.e. is
non-classical).

\emph{Decision Algorithms for Classical and Generalized-Classical States.}--- In the last portion of this letter we address the question of deciding whether a given multipartite state is classical or generalized-classical.  Our results, discovered independently, generalize the recent works on this topic \cite{Datta-2010a, Brodutch-2010a, Bylicka-2010a, Dakic-2010a, Rahimi-2010a} in which necessary and sufficient conditions have been provided for deciding the non-classical
bipartite states.  The techniques we use are similar to those in Ref. 
\cite{Dakic-2010a} in that both our algorithms involve checking commutation relations.  Interestingly, we find that deciding whether a state is generalized-classical reduces to a problem similar in nature to those well-studied in the field of signal processing \cite{Yeredor-2002a,Lathauwer-2008a}.  Hence, our use of semi-definite programming (SDP) in detecting generalized-classical states may be of interest to researchers in that subject, as well as the linear algebra community at large.  From a computational complexity perspective, our results expose the complexity contrast between deciding whether a state possesses entanglement,
which is NP-Hard \cite{Gurvits-2003a}, and deciding whether a state
possess non-classical correlations, which can be done in polynomial
time.

We first make the easy but important observation that it is no more
difficult to decide whether a state is fully generalized-classical (resp. fully
classical) than it is to decide if the state is generalized-classical (resp. classical) for just a single party.
\begin{lemma}
\label{lem:multipart} The state $\rho$ is fully generalized-classical
(resp. classical) if it is generalized-classical (resp. classical) for
all parties.
\end{lemma}
\begin{proof}
We will prove this for the bipartite case, but the idea immediately
generalizes to arbitrary number of parties.  Suppose $\rho=\sum_i
\rho_i \otimes \ketbra{b_i}{b_i} =
\sum_i\ketbra{a_i}{a_i}\otimes\sigma_i$ where the $\ket{b_i}$ and
$\ket{a_i}$ are linearly independent (resp. orthonormal).  Then we
see that each $\rho_i$ is a linear combination of the
$\ketbra{a_i}{a_i}$ so that $\ket{a_i}\otimes\ket{b_j}$ is a product
basis in which $\rho$ is diagonal.
\end{proof}

By Lemma \ref{lem:multipart}, it will be sufficient to only consider bipartite systems in the following discussion.  So introduce Alice and Bob and let $d_A$ and $d_B$ denote
the dimensions of their subsystems respectively.  Assume that some state $\rho$ is classical or generalized-classical with respect to Bob. By definition, there exists some basis
$\ket{b_i}$ such that
\begin{gather}
  \label{def:semiclassical}
  \rho= \sum_i p_i \r_i \ox \proj{b_i},
\end{gather}
while for classical states, the $\ket{b_i}$ are orthogonal.  Note
that in both cases, the contraction
$\langle\phi^{(A)}_1|\r|\phi^{(A)}_2\rangle$ will be diagonal in the basis
$\ket{b_i}$ for any two states
$\ket{\phi^{(A)}_1},\ket{\phi^{(A)}_2}\in\mathcal{H}_A$.  This fact leads to
the following theorem.
\begin{theorem}\label{thm:semiclassical}
  Let $\{\ket{\phi^{(A)}_i}\}$ be any orthonormal basis for $\mathcal{H}_A$.  Then $\rho$ is generalized-classical (resp. classical) if and only if
\begin{gather}\label{al:biclassical2}
  \rho_{ij}^{(B)}:=\langle\phi^{(A)}_i| \r |\phi^{(A)}_j\rangle
\end{gather}
is diagonal in the same (resp. orthonormal) basis $\{\ket{b_i}\}$
for all $i,j$.
\end{theorem}
\begin{proof}
  Necessity follows from the above observation.  For sufficiency, suppose that $\rho_{ij}^{(B)}=\sum_m b_{ijm}\ketbra{b_m}{b_m}$ where $\{\ket{b_m}\}$ is any linearly independent (resp. orthonormal) set spanning $\mathcal{H}_B$. From the general expansion $\rho=\sum_{ijmn} c_{ijmn}\ketbra{\phi^{(A)}_i} {\phi^{(A)}_j}\otimes\ketbra{b_m}{b_n}$, we see that $c_{ijmn}=\delta_{mn}b_{ijm}$ and so
\begin{equation}
\rho=\sum_{ijm}b_{ijm}\ketbra{\phi^{(A)}_i}{\phi^{(A)}_j}\otimes\ketbra{b_m}{b_m}
=\sum_m\rho_m\otimes\ketbra{b_m}{b_m}
\end{equation}
where $\rho_m=\sum_{ij}b_{ijm}\ketbra{\phi^{(A)}_i}{\phi^{(A)}_j}=\langle
b^\perp_m|\r|b^\perp_m\rangle$ and $\ket{b^\perp_m}$ are vectors
such that $\braket{b_i}{b^\perp_j}=\delta_{ij}$.  The last equation
implies that $\rho_m$ is semidefinite positive. Hence the state
$\rho$ is generalized-classical (resp. classical) as defined in
Eq.~\ref{def:semiclassical}. \end{proof}

Theorem~\ref{thm:semiclassical} implies that to decide whether
$\rho$ is generalized-classical for Bob, we need to check whether the
$\frac{1}{2}d_A(d_A-1)$ matrices $\{\bra{\phi^{(A)}_i}\r\ket{\phi^{(A)}_j}
\}_{1\leq i\leq j\leq d_A}$ of size $d_B\times d_B$ are
simultaneously congruent to diagonal matrices.  In a more general
form, this problem asks for some set $\{A_i\}_{i=0...m}$ of $n\times
n$ matrices whether there exists an invertible matrix $P$ such that
$PA_iP^\dagger=\Lambda_i$ is diagonal for all $i$.  This is a natural question to ask in linear algebra studies and we have already alluded to practical situations in which it arises outside of quantum information.  We thank Yaoyun
Shi for his assistance with the following.  To our knowledge, SDP is a previously unrecognized approach to solving the described problem.
\begin{lemma}\label{lem:SDP} $[\text{Shi}]$ Deciding if nonsingular $P$ exists such that $PA_iP^\dagger=\Lambda_i$ can be achieved by a semi-definite program (SDP).
\end{lemma}
To construct the algorithm, we first assume without loss of generality that the $A_i$ are hermitian.
For we can always write $A_i=A_i'+iA_i''$ where $A_i'$ and $A_i''$
are hermitian.  Then $PA_iP^\dagger$ is diagonal if and only if
$PA_i^\dagger P^\dagger$ is diagonal if and only if both
$PA_i'P^\dagger$ and $PA_i''P^\dagger$ are diagonal.  So with $A_i$
being hermitian, the $PA_iP^\dagger$ are hermitian and if
$PA_iP^\dagger=\Lambda_i$, the $PA_iP^\dagger$ are simultaneously
diagonalized and therefore $[PA_iP^\dagger,PA_jP^\dagger]=0$ for all
$i,j$.  Conversely, if this latter condition holds, then there
exists a unitary $U$ such that $UPA_iP^\dagger U^\dagger =\tilde{P}
A_i\tilde{P^\dagger} =\Lambda_i$ for all $i$.  So the question is
whether $PA_iP^\dagger PA_jP^\dagger=PA_jP^\dagger PA_iP^\dagger$
for all $i,j$.  Or in other words, $A_iWA_j=A_jWA_i$ where $W$ is a
positive-definite matrix.  Note that if $W$ is positive-definite,
then we can scale appropriately so that $W\geq I$.  Thus, we have
the SDP feasibility problem:
\begin{align}
&\text{Find} &&W\notag\\
&\text{subject to} &&A_iWA_j=A_jWA_i\;\text{for all}\;i,j&\notag\\
&&&W-I\geq 0.
\end{align}
Known algorithms based on the ellipsoid and interior-point methods
can efficiently solve this problem~\cite{Vandenberghe-1994a}.

To decide whether $\rho$ is classical for Bob, the situation is
easier.  We first begin by choosing any basis $\{\ket{\phi^{(A)}_i}\}$
for Alice and checking whether $\rho^{(B)}_{ij}$ is diagonalizable for
all $i,j$.  In total, there will be $\frac{d_A^2 - d_A}{2}$ matrices
to check.  If these are not diagonalizable, then by
Theorem~\ref{thm:semiclassical}, $\rho$ is not classical. If so,
$\rho$ is classical if and only if the commutation
$[\rho^{(B)}_{ij},\rho^{(B)}_{kl}]$ vanishes for all $i,j,k,l$, which
amounts to at most $\frac12 d_A^2(d_A^2-1)$ commutation relations to
check.  In the case that all operators commute, a common eigenbasis
$\{\ket{a_i}\}$ can be easily computed; the
sufficiency of Theorem~\ref{thm:semiclassical} proves $\rho$ to be
classical.

\emph{Physical detection of classical states}---Theorem
\ref{thm:semiclassical} can be experimentally implemented by a set
of projective operations and quantum state tomography.
A direct reconstruction of the elements in Eq. \ref{al:biclassical2} is not possible
since they are not Hermitian and therefore do not correspond to anything physical.  However, these terms
can be computed indirectly if Alice makes a set of
linearly independent projective operations (observables) that span
her Hilbert-Schmidt space:
$\mathcal{L}=\{\ket{\phi^{(A)}_i}\bra{\phi^{(A)}_i},
\ket{\psi^{(A)}_{ij}}\bra{\psi^{(A)}_{ij}},
\ket{\chi^{(A)}_{ij}}\bra{\chi^{(A)}_{ij}}\}$ where
$\ket{\chi^{(A)}_{ij}}=\frac{1}{\sqrt{2}}
\left(\ket{\phi^{(A)}_i}-i\ket{\phi^{(A)}_j}\right)$ and
$\ket{\psi^{(A)}_{ij}}=\frac{1}{\sqrt{2}}
\left(\ket{\phi^{(A)}_i}+\ket{\phi^{(A)}_j}\right)$ for $i>j$. With that we
have the elements of Eq. \ref{al:biclassical2}:
$\bra{\phi^{(A)}_i}\rho\ket{\phi^{(A)}_j} =\bra{\psi^{(A)}_{ij}}
\rho\ket{\psi^{(A)}_{ij}}
+i\bra{\chi^{(A)}_{ij}}\rho\ket{\chi^{(A)}_{ij}}-\frac{1+i}{2}
(\bra{\phi^{(A)}_i}\rho\ket{\phi^{(A)}_i}
+\bra{\phi^{(A)}_j}\rho\ket{\phi^{(A)}_j})$.

According to Theorem \ref{thm:semiclassical} a state $\rho$ is
classical if and only if it has the same orthonormal basis for
$\bra{\phi^{(A)}_i}\rho\ket{\phi^{(A)}_j}$ for all $i,j$. It is clear that
if $\mbox{Tr}_A[P\rho]$ is diagonal for all $P\in\mathcal{L}$ then
$\rho$ is classical.  Conversely, $\bra{\phi^{(A)}_i}\rho\ket{\phi^{(A)}_j}$
diagonal in some orthonormal basis for all $i,j$ implies that
$\mbox{Tr}_A[P\rho]$ is diagonal in same basis for all
$P\in\mathcal{L}$.  As the elements of $\mathcal{L}$ span Alice's
space, any POVM she can perform will have operator elements with
each being a linear combination of these projectors.  Furthermore, if we consider ``Alice's'' system as the joint system of $N-1$ parties, then any local POVM performed by the $N-1$ parties will have product operators $E_{\vec i}=\bigotimes_{j=1}^{N-1}E_{i_j}^{(j)}$ also being a linear combination of projectors from $\mathcal{L}$, and conversely any element of $\mathcal{L}$ can be expressed as a linear combination of product operators constituting complete local measurements on the $N-1$ subsystems.  Thus we obtain the following: 
\begin{theorem}
\label{lem:classical commutator}
An $N$-partite state $\rho$ is classical with respect to party $k$ if and only if
for any local POVM performed by the other parties,
\begin{gather}
[\rho_{\vec{i}},\rho_{\vec{i'}}]=0\;\;\mbox{for all}\;\;\vec{i},\vec{i'}
\end{gather}
where $\rho_{\vec{i}}=\mbox{Tr}_{\overline k}[E_{\vec{i}}\rho]$. 
\end{theorem}
Quantum mechanics and
commutators are intimately related since days of the theory's
foundation. Here, we see that the non-classical nature of a state can be detected precisely by the non-commutativity of reduced states after some local POVM is locally implemented on all but one of the subsystems.

\emph{Conclusion.}---We have introduced a class of states called generalized-classical which permit the purely classical task of non-disruptive local state identification when general quantum measurements are used.  In this sense, generalized-classical states can be said to hover near the quantum-classical boundary.  We have provided methods, both analytic and physical, which decide if a state is classical or generalized-classical.  For the latter, our algorithm amounts to a seemingly novel way for deciding whether a set of matrices can be simultaneously diagonalized by a general (non-necessarily orthogonal) congruence transformation.  Our results hold in the multipartite setting where states can be classical or generalized-classical with respect to one or many of the involved parties.  
We believe these results are helpful in better understanding the intersection between classical and quantum regimes.

{\bf Acknowledgment.} LC thank Dr. Ying Li and Prof. Wei Song for
helpful discussions. KM thanks B. Dakic, C. Rodriguez-Rosario, 
and V. Vedral for discussions. EC is partially supported by the 
U.S. NSF under Awards 0347078 and 0622033. The Center for Quantum 
Technologies is funded by the Singapore Ministry of Education and 
the National Research Foundation as part of the Research Centres 
of Excellence programme.

\bibliography{ClassicalStatesBib}

\end{document}